\documentclass{article}
\usepackage[final]{neurips_2019}

\usepackage{algorithm}
\usepackage[noend]{algpseudocode}
\usepackage{amsthm}
\usepackage{amsmath}
\usepackage[pdftex]{graphicx}
\usepackage{epstopdf}
\usepackage{xcolor}

\definecolor{seccolor}{rgb}{0,0.1,0.4}

\newtheoremstyle{mystyle}{}{}{\em}{}{\color{seccolor}\bfseries}{.}{ }{}

\theoremstyle{mystyle}
\newtheorem{defn}{Definition}
\newtheorem{thm}{Theorem}
\newtheorem{lmm}[thm]{Lemma}
\newtheorem{cor}[thm]{Corollary}

\newcommand{\parb}[1]{\Phi_{#1}\left(x\right)}
\newcommand{\range}[1]{\mathcal{R}_{#1}}

\author{
Mihir Sahasrabudhe\\
Centre de Vision Num\'{e}rique \\CentraleSup\'{e}lec\\ Universit\'{e} Paris-Saclay \\
91190 Gif-sur-Yvette, France \\
\texttt{mihir.sahasrabudhe@ecp.fr} 
\And
Siddhartha Chandra\\
Centre de Vision Num\'{e}rique \\CentraleSup\'{e}lec\\ Universit\'{e} Paris-Saclay \\
91190 Gif-sur-Yvette, France \\
\texttt{robinchandra19@gmail.com}
}

\title{Proof of Correctness and Time Complexity Analysis of a Maximum Distance Transform Algorithm}

\begin{document}

\maketitle

\begin{abstract}
The distance transform algorithm is popular in computer vision and machine 
learning domains. It is used to minimize quadratic functions over a grid of points.
Felzenszwalb and Huttenlocher 
describe an $O(N)$ algorithm for computing the {\em minimum} distance transform 
for quadratic functions. Their algorithm works by computing the lower 
envelope of a set of parabolas defined on the domain of the function. In this 
work, we describe an {\em average time} $O(N)$ algorithm for ``maximizing'' 
this function by computing the {\em upper envelope} of a set of parabolas.
We study the duality of the minimum and maximum distance transforms, give a 
correctness proof of the algorithm and its runtime, and discuss potential 
applications.
\end{abstract}


\section{Introduction}
The distance transform algorithm is frequently used for solving graph inference 
problems in computer vision, machine learning and other domains, as well as in 
image processing applications. It has gained popularity in recent years with 
the proliferation of works based on deformable part models 
(DPMs)~\citep{dpmOrig,Zhu_facedetection}. 

The {\em minimum} distance transform algorithm is used to minimize functions of 
specific forms on $1D$ or $2D$ grids. In this work, we assume the form of the function to be 
quadratic. This assumption is a design choice made by most deformable part 
models~\citep{dpmOrig,Zhu_facedetection}. More precisely, the minimization 
problem on a $1D$ grid can be expressed as

\begin{equation}
D(x) = \min_{p}  I(p) + \alpha(p-x)^2 + \beta(p-x)\,.
\label{eqn:distancetransform1D}
\end{equation} 

The problem on a $2D$ grid is given by the equation
\begin{equation}
D(x,y) = \min_{p,q}  I(p,q) + \alpha(p-x)^2 + \beta(p-x) + \gamma(q-y)^2 + 
\delta(q-y)\,.\label{eqn:distancetransform2D}
\end{equation}

Conventionally, distance transforms have been used for minimizing functions. 
However, maximizing these functions is an equivalent problem. We 
discuss this duality in section \ref{section:duality}. We define {\em maximum} 
distance transforms by the following equations for the $1D$ and $2D$ cases 
respectively:

\begin{align}
D(x) &= \max_{p}  I(p) + \alpha(p-x)^2 + \beta(p-x)\,;\text{~and} \label{eqn:maxdistancetransform1D} \\
D(x,y) &= \max_{p,q}  I(p,q) + \alpha(p-x)^2 + \beta(p-x) + \gamma(q-y)^2 + 
\delta(q-y)\,. \label{eqn:maxdistancetransform2D}
\end{align}

Kindly note that Euclidean distance transforms are a special case 
($\alpha,\gamma=1; \beta,\delta=0$) of the problems in equations 
\ref{eqn:distancetransform1D}-\ref{eqn:maxdistancetransform2D}. These quadratic functions
can be seen as parabolas centered at the grid points, as is shown in Figure 
\ref{fig:complete}(a). Each grid point $p$ in equation 
\ref{eqn:distancetransform1D} is associated with a parabolic function of the 
following form: 
\begin{equation}
\parb{p} = I(p) + \alpha(p-x)^2 + \beta(p-x)\,.
\label{eqn:parabola}
\end{equation}
Therefore, finding the minimum distance transform can be understood as finding 
the minimum out of these functions at each point $x$ in the domain. In short, 
the minimum distance transform is given by the lower envelope of these 
parabolas. Equivalently, the maximum distance transform is given by the upper 
envelope of these parabolas. We will henceforth discuss the maximum distance transform in this paper. We also use the term \emph{upper envelope} to describe the curve defining the maximum distance transform for all points in the domain. 

We now define the \emph{range} of a parabola as follows.
\begin{defn}[Range of a parabola]
  If $\parb{p}$ is a parabola in the upper envelope centred at grid point 
$p$, its range, $\range{p}$, is the interval 
$(\sigma_1,\sigma_2]$ in which $\parb{p}$ forms the upper envelope.
  \label{defn:range}
\end{defn}


\begin{figure}[t!]
 \centering
 \includegraphics[scale=0.62]{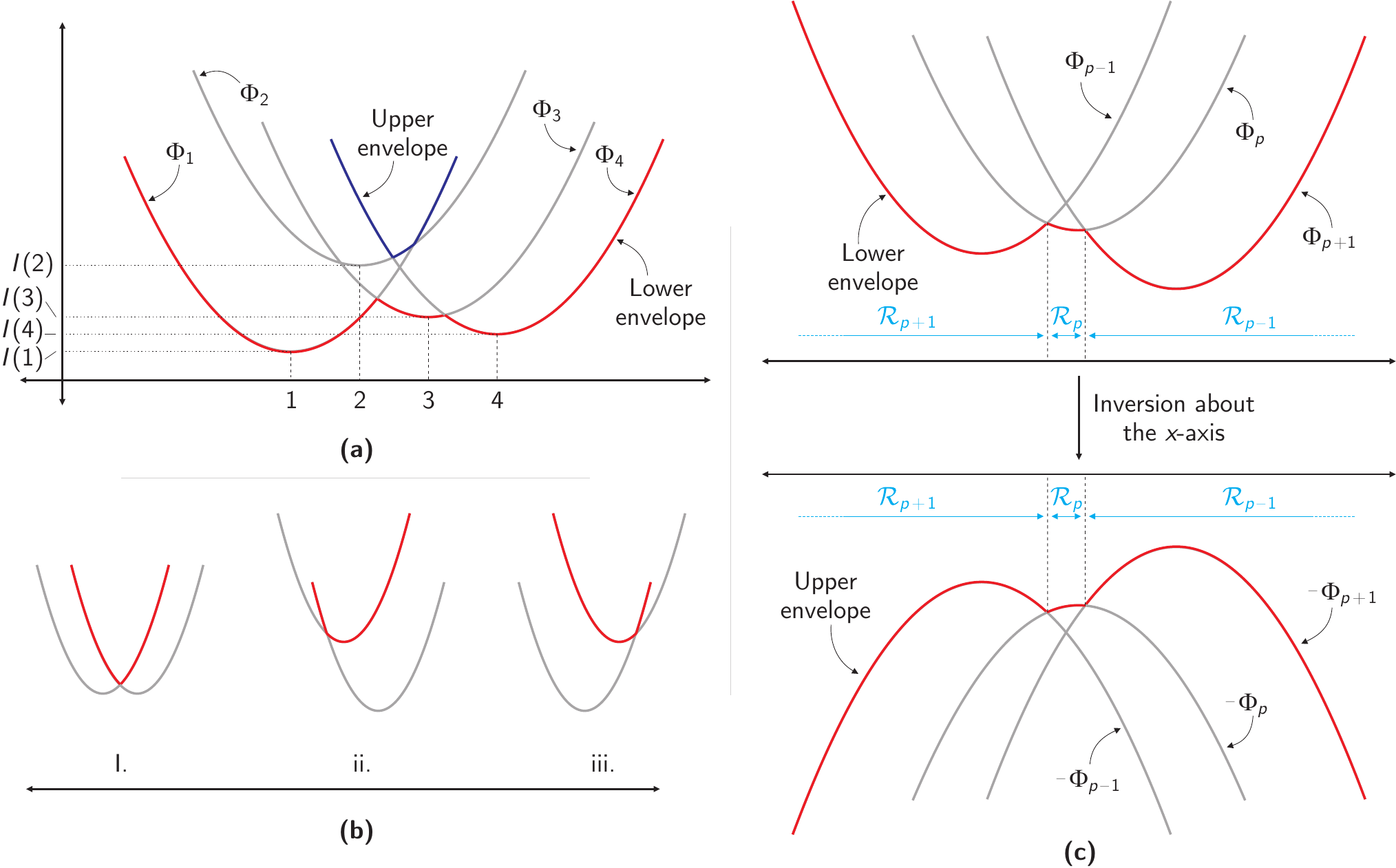}
 \caption{(a) Upper and lower envelopes of a set of parabolas: each grid point has 
one parabola associated with it. The minimum distance transform is given by the 
lower envelope whereas the maximum distance transform is given by the upper 
envelope.\hspace{\textwidth}
(b) The upper envelope at the intersection of two parabolas. The parabola on the right becomes the upper envelope at the left of the point of intersection. The parabola on the left becomes the upper envelope at the right of the point of intersection.\hspace{\textwidth}
(c) The duality between minimum and maximum distance transforms: 
inverting the signs of the quadratic terms changes upward opening parabolas to 
downward opening parabolas; If we know how to compute the lower envelope for 
one case, we can compute the upper envelope for the other case.}
 \label{fig:complete}
\end{figure}

The ranges of several parabolas are shown in Figure \ref{fig:complete}(c).
Felzenszwalb and Huttenlocher \citep{Felzdt} describe an $O(N)$ algorithm for 
computing the minimum distance transform for Equations 
\ref{eqn:distancetransform1D} and \ref{eqn:distancetransform2D}.  Their 
algorithm works by finding the lower envelope of these parabolas. In this work, 
we describe an $O(N)$ algorithm for computing the maximum distance transform 
(Equations \ref{eqn:maxdistancetransform1D} and 
\ref{eqn:maxdistancetransform2D}). Our algorithm works by finding the upper 
envelope of the parabolas.

We begin by showing the duality between maximum and minimum distance 
transforms, followed by our algorithm in detail. We then give a correctness 
proof and analyze the complexity of our algorithm. Finally, we discuss 
applications of our algorithm.

\section{Duality between Minimum and Maximum Distance Transforms}
\label{section:duality}

The duality between the Minimum and Maximum distance transforms can be trivially 
explained by the equation

\begin{equation}
\max_{p}  I(p) + \alpha(p-x)^2 + \beta(p-x) = \min_{p}  -I(p) - \alpha(p-x)^2 - 
\beta(p-x)\,.
\label{eqn:duality}
\end{equation}

Thus, if we know how to solve the minimum distance transform we can also solve 
the maximum distance transform simply by changing the sign of the function. 
The algorithm in \citep{Felzdt} solves the minimum distance 
transform for {\em upward opening parabolas}. Upward opening parabolas are 
characterized by positive quadratic terms, or more precisely $\alpha,\gamma>0$ 
(Equations \ref{eqn:distancetransform1D}-\ref{eqn:maxdistancetransform2D}). 
Thus, the algorithm in \citep{Felzdt} can be used to find the upper envelope / 
maximum distance transform when $\alpha,\gamma<0$.

This is better illustrated in Figure \ref{fig:complete}(c). Upon inverting the 
signs of the quadratic terms, the {\em upward opening parabolas} become {\em 
downward opening parabolas}: the algorithm which gave the lower envelope now 
gives us the upper envelope. However, as mentioned before, the algorithm in 
\citep{Felzdt} solves the minimum distance transform only for upward opening 
parabolas. In the next section, we describe an 
algorithm which finds the upper envelope for upward opening parabolas, and 
therefore can be used to find the lower envelope of downward opening parabolas.

\section{Algorithm}
In this section we describe an algorithm for computing the maximum distance 
transform. We begin with the $1D$ grid case (Equation 
\ref{eqn:maxdistancetransform1D}), and extend it for arbitrary dimensions. Our 
algorithm begins by computing the upper envelope of the parabolas. The distance 
transform at a point $x$ is simply the height of the upper envelope at $x$. 
This is illustrated in Figure \ref{fig:complete}(a).

\subsection{Intersection of two parabolas}
Kindly note that all parabolas, albeit centered at different points in space, 
have the same shape. This is due to the fact that they have the same $\alpha$ 
parameter. The parabola at a grid point $p$ is defined in Equation 
\ref{eqn:parabola}.

 \begin{lmm}
  Parabolas $\parb{p}$ and $\parb{q}$ at two grid points $p$ and $q$, respectively, 
intersect at exactly one point.
  This point of intersection is given by 
  \begin{equation}
   s_{p,q} = \frac{(I(q) + \alpha q^2 + \beta q) - (I(p) + \alpha p^2 + \beta 
p)}{2\alpha(p - q)}\,.
  \end{equation}
  \label{lmm:intersection}
 \end{lmm}
 \begin{proof}
  The proof follows from algebra. The point of intersection of the parabolas 
$\parb{p}$ and $\parb{q}$ can be obtained by setting $\parb{p} = \parb{q}$ . Since
  the coefficient of $x^2$ in both of them is independent of $p$ and $q$, it 
cancels out, giving us a unique solution. 
 \end{proof}

\subsection{Upper Envelope}

  \begin{lmm}
  If $p < q$ are two grid points with $s_{p,q}$ being the point of intersection 
of $\parb{p}$ and $\parb{q}$, 
  then $\parb{p} > \parb{q}$ for $x > s_{p,q}$, and $\parb{p} < \parb{q}$ for $x < s_{p,q}$.
  \label{lmm:p_val}
 \end{lmm}
 \begin{proof}
  To prove this, we carefully examine the three conditions in which two 
parabolas $\parb{p}$ and $\parb{q}$ can intersect. Each parabola $\parb{p}$ is 
centered at the grid point $p$. The \emph{left arm} of a parabola is defined to 
be the part of the parabola to the left of $p$. Similarly, the \emph{right arm} 
of the parabola is defined to be the part of the parabola to the right of $p$. 
There are three possibilities: 
  \begin{enumerate}
   \item {\bfseries The right arm of $\parb{p}$ intersects the left arm of 
$\parb{q}$.}~~This corresponds to $p < s_{p,q} < q$. Since $\parb{p}$ and $\parb{q}$ have positive and negative gradients, 
respectively, at $s_{p,q}$, the Lemma is trivially satisfied. See Figure \ref{fig:complete}(b), part i.
   \item {\bfseries The left arm of $\parb{p}$ intersects the right arm of 
$\parb{q}$.} This situation can never occur since $p < q$, and the left arm of 
$\parb{p}$ lies entirely to the left of $p$ whereas the right arm of $\parb{q}$
         lies entirely to the right of $q$.
   \item {\bfseries The left (right) arm of $\parb{p}$ intersects the left 
(right) arm of $\parb{q}$.}~~~This corresponds to $s_{p,q} \leq p < q$ and $p < q \leq s_{p,q}$, respectively. The expression for the gradient of the parabola $\parb{p}$ is given by
\begin{equation}
	\frac{d}{d x} \parb{p} = -2\alpha (p-x) -\beta
\end{equation}
The gradient of $\parb{p}$ at $x$, therefore, depends on $p-x$. 
In this case, since $|s_{p,q} - p|$ is less (greater) 
than $|s_{p,q} - q|$,
         the absolute value of the gradient of $\parb{p}$ at $s_{p,q}$ will be less (greater) 
than that of $\parb{q}$. It follows that $\parb{p} > \parb{q}$ for $x > s_{p,q}$, and $\parb{p} < \parb{q}$ for $x < s_{p,q}$. See Figure \ref{fig:complete}(b), parts ii, iii.
  \end{enumerate}
  The proof of the Lemma follows from the three exhaustive cases. Figure \ref{fig:complete}(b) gives visual confirmation to this Lemma.
 \end{proof}

 For simplicity, we will define an order relation over the ranges of parabolas. 
 \begin{defn}[$\succ$ and $\prec$ relations on ranges of parabolas]
 We say $\range{p} = (\sigma_1, \sigma_2] \prec \range{q}$ if and only if $\sigma_2 \leq \inf\left\{x | x \in \range{q}\right\}$. $\succ$ is defined to be the inverse relation to $\prec$ so that $\range{p} \prec \range{q} \Leftrightarrow \range{q} \succ \range{p}\,$. Intuitively, $\range{p} \prec \range{q}$ signifies that $\range{p}$ \emph{lies completely to the left} of $\range{q}\,$.
 \end{defn}
 \begin{defn}[Adjacency of ranges]
  Two ranges $\range{p}$ and $\range{q}$ are said to be \emph{adjacent} if and only if $\sup\left\{x | x \in \range{1}\right\} = \inf\left\{x | x \in \range{2}\right\}$, where $\range{1}$ and $\range{2}$ are $\range{p}$ and $\range{q}$ in no particular order. 
 \end{defn}
 
 \begin{lmm}
  If $\parb{p} \neq \parb{q}$ are two parabolas in the upper envelope such
that $\range{p}$ and $\range{q}$ are adjacent, with $\range{p} \prec \range{q}$, then $p > q$, and vice versa.
  \label{lmm:adj_interval}
 \end{lmm}
 \begin{proof}
 Since $\range{p}$ and $\range{q}$ are adjacent and $\range{p} \prec \range{q}$, we can assume their ranges are $\range{p} = (\sigma_1, s_{p,q}]$, and $\range{q} = (s_{p,q}, \sigma_3]$, for some $\sigma_1$ and $\sigma_3\,$.
If $p < q$, Lemma \ref{lmm:p_val} is violated. Hence, $p > q$. Conversely, if $p > q$ and $\parb{p}$ and $\parb{q}$ have adjacent ranges in the upper envelope, then $\range{p} \prec \range{q}$ from Lemma \ref{lmm:p_val}.
 \end{proof}

 \begin{cor}
  (to Lemma \ref{lmm:adj_interval}) If a parabola $\parb{p}$ is maximum over 
$\range{p} = (\sigma_1, \sigma_2]$,
  $\sigma_2$ being finite, then $\exists q < p$ such that $\range{q} \succ \range{p}$. 
  \label{cor:finite_range}
 \end{cor}
 \begin{proof}
  The finiteness of $\sigma_2$ guarantees the presence of another parabola 
which is maximum over an interval in $(\sigma_2, \infty)\,$. Lemma \ref{lmm:adj_interval} tells 
us that a parabola $\parb{q}$ that is maximum over an interval $\range{q}$ adjacent to $\range{p}$ with $\range{p} \prec \range{q}$ implies $q < p$. This
  guarantees the existence of one such $q$. Further, applying this argument 
  repeatedly, we conclude that for all parabolas $\parb{t}$ whose ranges satisfy $\range{t} \succ \range{p}$, $t < p\,$. 
 \end{proof}
 
 \begin{cor}
 (To Lemma \ref{lmm:adj_interval}) If $\range{p}$ and $\range{q}$ are ranges of two parabolas in the upper envelope such that $\range{p} \succ \range{q}$, then $p < q$, and vice versa. 
 \label{cor:parabola_ordering}
 \end{cor}
 \begin{proof}
  We know that $\range{p} \succ \range{q}$. Let $t$ be the parabola in the upper envelope so that $\range{t} \prec \range{p}$ and is adjacent to it. By Lemma \ref{lmm:adj_interval}, $p < t$. If we keep applying this argument repeatedly to $t$, we will eventually reach $q$. The proof follows. 
  
  Conversely, if $p < q$ and $\parb{p}$ and $\parb{q}$ are part of the upper envelope, then $\range{p} \succ \range{q}$. by a similar repeated application of Lemma \ref{lmm:adj_interval}. 
 \end{proof}

\subsection{Algorithm}
Before describing the algorithm, we describe the data-structures employed.
 \begin{enumerate}
  \item $k$ is an integer counting the number of parabolas in the upper 
envelope. 
  \item $v[\cdot]$ is an array which holds the indices of the parabolas 
currently forming the upper envelope. At any stage of the algorithm $v[\cdot]$ 
holds $k+1$ elements.
  \item $z[\cdot]$ is an array holding the ranges of the parabolas in \emph{decreasing} 
order, i.e. from $+\infty \rightarrow -\infty$. Thus, the range of 
$p$-th parabola, given by $\parb{v[p]}$ is given by $\range{v[p]} = (z[p+1], z[p]]$. 
 \end{enumerate}

We assume that all parabolas are ordered according to the horizontal locations 
of their grid points.
We start by initialising the upper envelope to be the parabola at the first 
grid point. The range is initialised to be $(+\infty,-\infty)$.
We compute the upper envelope by iteratively scanning the grid points from 
left to right. Each time we encounter a new grid point $q$, we compute its 
intersection with the existing parabolas in the upper envelope. If $\parb{q}$ intersects 
a parabola $\parb{p}$ in the upper envelope inside $\range{p}$, we update the upper envelope: all parabolas at grid points 
to the right of $p$ are removed from the upper envelope, the parabola $q$ is 
added to the upper envelope, and the ranges of $p$ and $q$ are updated. It is guaranteed that $\parb{q}$ will intersect with at least one parabola $\parb{p}$ inside $\range{p}$.
This is demonstrated in Figure \ref{fig:new_par} and formally proven in Lemma \ref{lmm:env_mod}.

After computing the upper envelope, we scan the grid points from left to 
right filling in the values of the distance transform by investigating the 
range array. A pseudo-code is described in Algorithm \ref{alg:dt}. 
We now prove the correctness of Algorithm \ref{alg:dt}. 

 \begin{algorithm}[h!]
  \caption{Function maximumDistanceTransform}
  \label{alg:dt}
  \begin{algorithmic}[1]
   \State \textbf{Input: The unary function $I$, parameters $\alpha, \beta$, set of grid points $\{0, 1, \ldots, N\}$}
   \State Initialise arrays: $v$ and $z$. 
   \State $v[0] \gets 0$
   \State $z[0] \gets \infty$
   \State $z[1] \gets -\infty$
   \State $k \gets 0$
   \For{$q = 1 \rightarrow N - 1$} \label{alg:dt:1for} 
    \For{$p = 0 \rightarrow k$}				 \label{alg:dt:2for}
     \State $s \gets \textstyle{\frac{(I(q) + \alpha q^2 + \beta q) - (I(v[p]) + \alpha v[p]^2 + \beta v[p])}{2\alpha(v[p] - q)}}$ \label{alg:dt:intersect}
     \If{$s > z[p + 1] \mathbf{~and~} s \leq z[p]$}   \label{alg:dt:if}
      \State $k \gets p + 1$				 \label{alg:dt:new_p_s}
      \State $v[k] \gets q$
      \State $z[k + 1] \gets -\infty$
      \State $z[k] \gets s$
      \State break							 
\label{alg:dt:new_p_e}
     \EndIf
    \EndFor
   \EndFor
   
   \For{$q = 0 \rightarrow N - 1$} \label{alg:dt:3for}
    \While{$z[k] < q$}
     \State $k \gets k - 1$
    \EndWhile
	\State $DT[q] \gets I[v[k]] + \alpha(v[k]-q)^2 + \beta(v[k]-q)$
   \EndFor
  \end{algorithmic}
 \end{algorithm}

 \begin{thm}
  Algorithm 1 correctly computes the maximum distance transform.
  \label{thm:correctness}
 \end{thm}
 \begin{proof}
  We use the principle of mathematical induction on the number of parabolas, $N$, for our proof. 
  
  \textbf{Base case.} For the base case, $N = 1$, the algorithm does not enter the loop on line \ref{alg:dt:1for} at all.
  Since the arrays $v[\cdot]$ and $z[\cdot]$ are already initialised accordingly, the algorithm returns the correct envelope---a single parabola having the range $(-\infty, \infty)\,$.
  
  \textbf{Inductive step.} Let us assume that the algorithm gives the correct 
upper envelope for the
  $q - 1$ grid-points considered so far. We now consider the $q$-th parabola. 
  The algorithm then computes
  $s_{v[p],q}$---the point of intersection 
  of the $p$-th parabola in the envelope, $\parb{v[p]}$, with $\parb{q}$. (Line 
\ref{alg:dt:intersect} in Algorithm \ref{alg:dt}). Recall 
  $\range{p}$ from Definition \ref{defn:range}. 
  There are two possibilities:
  \begin{description}
   \item [\bfseries Case 1: $s_{v[p],q} \in \range{v[p]}$.] It follows from Lemma \ref{lmm:p_val} that $\parb{q} > \parb{v[p]}\, \forall\,x < s_{v[p],q}$, and subsequently, $\parb{q} > 
\parb{t}$ in the same
   interval $\forall t < q$. 
   Hence, $\range{q} = (-\infty, s_{v[p],q}]$. Furthermore, parabolas
   in the upper envelope whose ranges fall in $\range{q}$ are no longer
   maximum in their respective ranges. These parabolas correspond to 
   all the grid points between $v[p]$ and $q$, and they are removed 
   by readjusting the value of $k$ in the algorithm. 
   Furthermore, no more parabolas can be added or removed from the envelope in 
this iteration,
   so we break the \texttt{for} loop. Lines \ref{alg:dt:new_p_s} to 
\ref{alg:dt:new_p_e} in 
   Algorithm \ref{alg:dt} handle precisely this case. 

   \item [\bfseries Case 2: $s_{v[p],q} \notin \range{v[p]}$.] This says 
that the point of intersection of $\parb{v[p]}$ and 
   $\parb{q}$ lies outside $\range{v[p]}$. In this case, the presence of 
$\parb{q}$ in no way affects $\range{v[p]}$.
   To see why, let us say $\range{v[p]} = (\sigma_1, \sigma_2]$. 
   If $s_{v[p],q} \le \sigma_1$, then $\parb{q} < \parb{v[p]} \forall x > 
s_{v[p],q}$ 
   from Lemma \ref{lmm:p_val}, and $\parb{q} < \parb{v[p]}, \,\forall x \in \range{v[p]}$. 
   
   On the other hand, $s_{v[p],q} > \sigma_2$ is impossible (assuming $\sigma_2$ 
   is finite). We will show this by contradiction. Let us say $s_{v[p],q} > 
\sigma_2$ holds. 
   It follows from Corollary \ref{cor:finite_range} that $\exists m < v[p]$ 
   such that $\parb{m} > \parb{v[p]} \forall x \in (\sigma_2, s_{v[p],q})$ and 
$s_{m,q} \in \range{m}$. 
   Now the algorithm scans parabolas in $v[\cdot]$ (in the second for loop) from left to right, and so 
the
   parabola $\parb{v[p]}$, which comes in the list after $\parb{m}$, would be removed from $v[\cdot]$ due to the intersection of $\parb{m}$ and $\parb{q}$
   (line \ref{alg:dt:if} in Algorithm \ref{alg:dt}). This is a contradiction as this implies the intersection of $\parb{v[p]}$ and $\parb{q}$ is never computed.  

   As this requires no modifications to the ranges of any parabolas in the 
envelope or the envelope itself, 
   algorithm \ref{alg:dt} does nothing if $s_{v[p],q} \notin 
\range{v[p]}$ holds. 
  \end{description}
  
  The second \texttt{for} loop (line \ref{alg:dt:2for}) compares a new 
parabola with existing 
  parabolas in $v[\cdot]$. The discussion above shows that this iteration 
  adds the $q$-th parabola to the already existing upper envelope.
  It follows from the preceeding analysis that this addition is indeed done correctly. A parabola
  is added only when the condition in case 1, above, is satisfied. By
  Lemma \ref{lmm:p_val}, we know that the new parabola does not dominate the current upper envelope $\forall x \in (s_{v[p],q},\infty)$, and
  hence doesn't affect the part of the envelope in this range. 
 \end{proof}

 The following lemma follows as a result of the proof of Theorem 
\ref{thm:correctness}. 
 \begin{lmm}
  Each parabola, when first considered, becomes part of the upper envelope.
  \label{lmm:env_mod}
 \end{lmm}
 \begin{proof}
  The algorithm initialises the limits of the upper envelope to $(-\infty, 
\infty)$. Throughout the algorithm, these limits are never changed. Hence, a new parabola will
  always intersect the upper envelope at some point. Therefore, each parabola being considered will always be added to the upper envelope (Figure \ref{fig:new_par}).
  It may cause the deletion of previously considered parabolas from the upper envelope, and it may be deleted from the upper envelope by a
  subsequent parabola. Thus, in each outer loop over $q$ (line \ref{alg:dt:1for}) the new parabola being considered modifies the upper envelope.
 \end{proof}
  
  We can state some further properties of the upper envelope in the following Lemmas. 
 \begin{lmm}
  The first grid point is always a part of the upper envelope. 
  \label{lmm:first_grid_point}
 \end{lmm}
 \begin{proof}
  Let $p > 0$ be the left-most grid point included in the upper envelope (assuming grid points start at $0$). 
  Since there are no grid points $q < p$ which form a part of the upper envelope, Corollary \ref{cor:parabola_ordering} tells us that $\range{p}$ must go till $+\infty$. However, as as $0 < p$, $\parb{0} > \parb{p}\, \forall x > s_{0,p}$ (from Lemma \ref{lmm:p_val}). This implies $\range{p}$ cannot go on until $+\infty$. Since this argument breaks down only when $p = 0$, we conclude that the first grid point must be included in the upper envelope. 
 \end{proof}
 \begin{lmm}
  The first and last grid points are always part of the upper envelope.
  \label{cor:extreme_grid_points}
 \end{lmm}
 \begin{proof}
  The proof for the first grid point follows from Lemma \ref{lmm:first_grid_point}. Furthermore, note that each time a new parabola is scanned, it is ``temporarily'' the last parabola, until the remaining are scanned as well. Since the last parabola scanned will always be a part of the upper envelope even if there are parabolas after it which can potentially remove it from the upper envelope, the proof is complete.
 \end{proof}

 \begin{figure}[t!]
  \centering
  \includegraphics[scale=0.62]{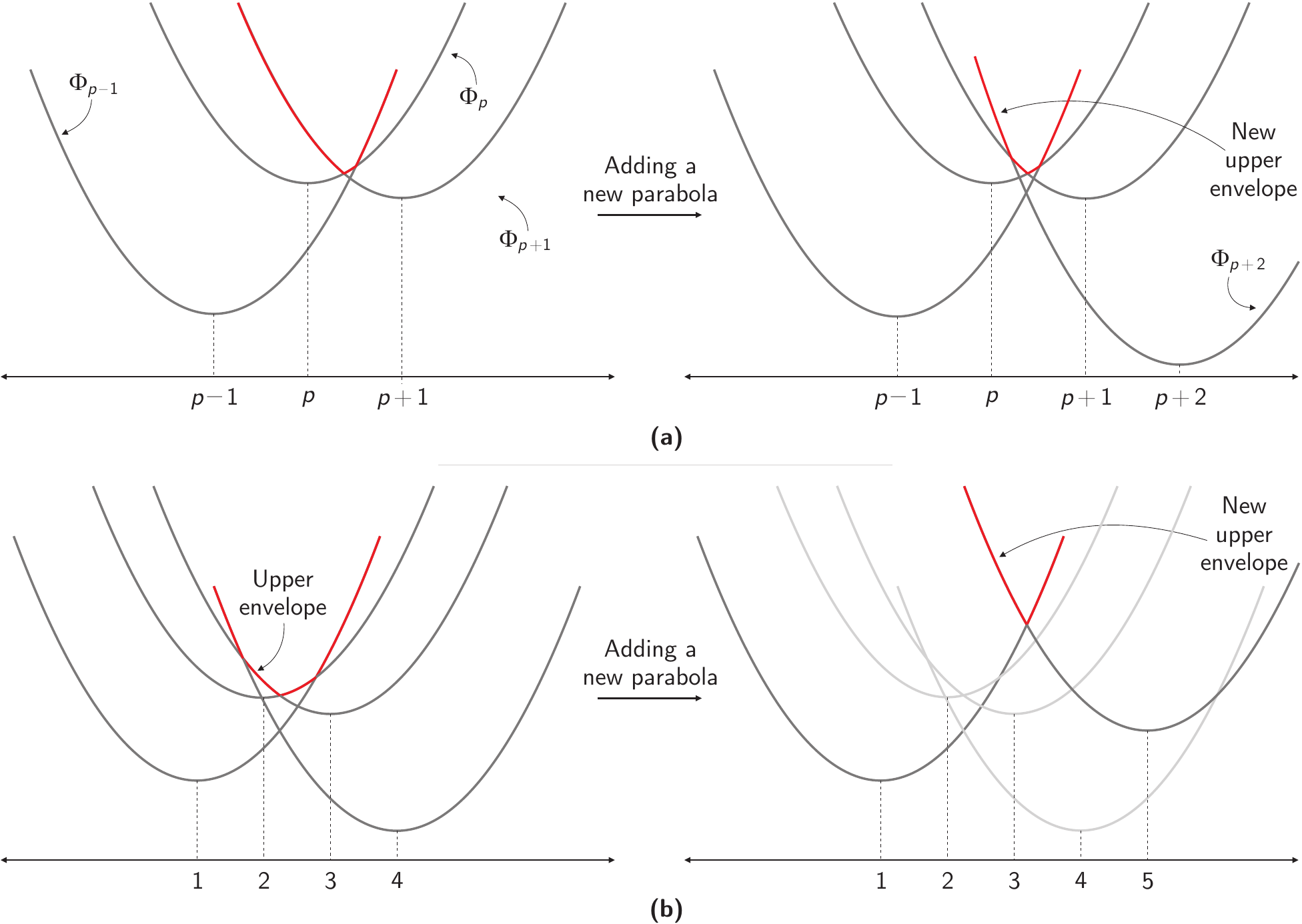}
  \caption{Each parabola when first considered is added to the upper envelope (Lemma \ref{lmm:env_mod}): (a) shows the case in which the new parabola
  is added to the upper envelope without removing any existing parabolas, (b) shows the case where the addition of the new parabola causes the deletion of
  several parabolas from the upper envelope.}
  \label{fig:new_par}
 \end{figure}
\subsection{Runtime Complexity Analysis}
 \begin{figure}[ht!]
  \centering \includegraphics[scale=0.7]{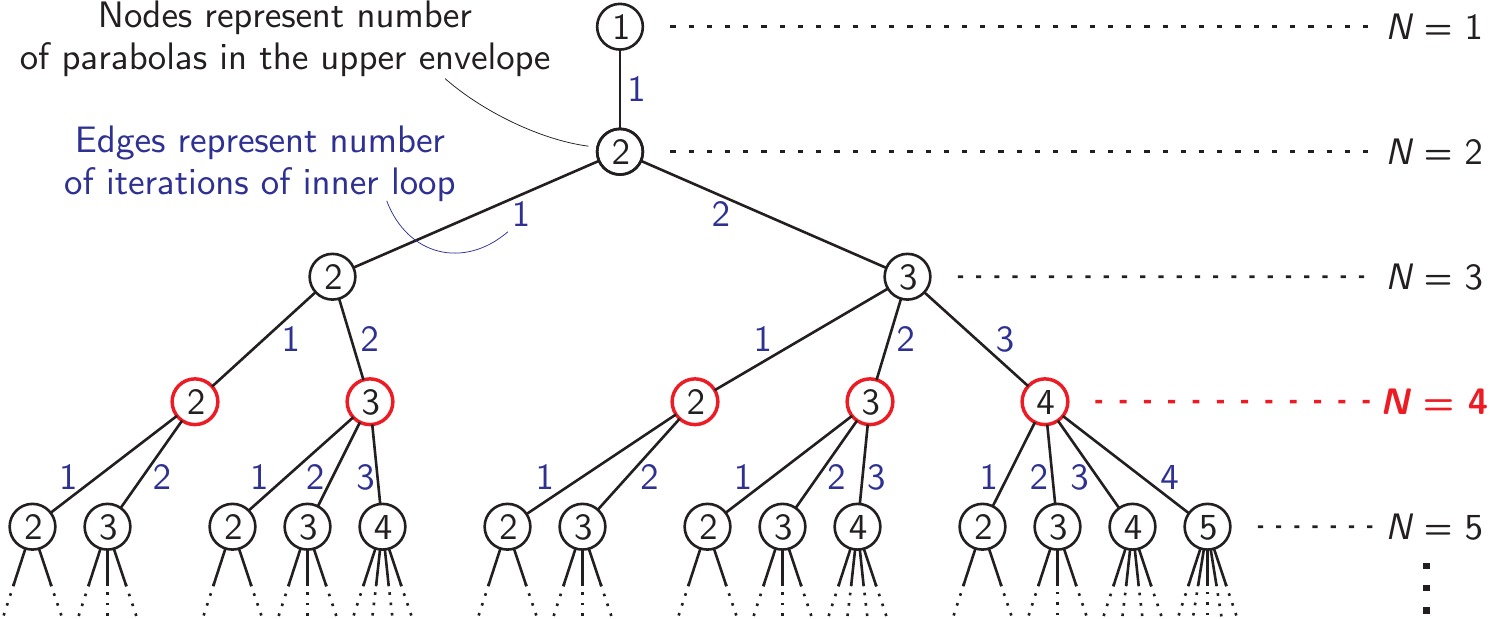}
  \caption{Tree construction to enumerate all possible number of iterations in the inner loop of the algorithm (line \ref{alg:dt:2for}), and the number of parabolas in the upper envelope.
  The nodes represent the number of parabolas in the upper envelope, the edge weights indicate the number of inner loop iterations.
  Each level in the tree represents an outer loop (line  \ref{alg:dt:1for}). At the $N$-th level of the tree, $N$ grid points have already been examined by the algorithm, therefore the number of parabolas in the upper envelope should lie between $2$ and $N$.}
  \label{fig:tree}
 \end{figure}
 To ascertain the runtime complexity of our algorithm, we consider lines 
\ref{alg:dt:1for},\ref{alg:dt:2for},\ref{alg:dt:3for} of algorithm 
\ref{alg:dt}. The loop in line \ref{alg:dt:1for} iterates once over each of the 
$N-1$ grid points. The inner loop (line \ref{alg:dt:2for}) iterates over the 
parabolas in the upper envelope until the condition in line \ref{alg:dt:if} is 
satisfied. In the {\textbf{worst case}}, the inner loop (line 
\ref{alg:dt:2for}) would always iterate through all the parabolas in the 
upper-envelope, which are bounded by $q$ in line \ref{alg:dt:1for}. This is a 
strict upper bound because the number of parabolas in the upper envelope will 
always be less than or equal to the number of grid points we have considered at 
any iteration. Therefore, the sum of inner loop iterations over the 
$N-1$ outer loop iterations is $1 + 2 + 3 + \cdots + N-1 = 
\frac{N\times(N-1)}{2}$. The loop in line \ref{alg:dt:3for} iterates once over 
each grid point $N$. The worst case complexity is therefore $O(N^2+N) = O(N^2)$.
 
Even though the worst case runtime complexity of the algorithm is the same as that of 
the brute force solution, it is significantly faster in practice. We now show 
that the average case runtime complexity of the algorithm is $O(N)$. To achieve this, 
we enumerate all possibilities that the inner loop in line \ref{alg:dt:2for} of 
the algorithm will see. Consider the tree in Figure \ref{fig:tree}. The nodes 
in the tree represent the number of parabolas in the upper envelope at any 
iteration. The edges represent the number of iterations that the inner loop 
iterates for. For instance, at the beginning of the first outer loop $N=1$, we 
have a single parabola in the upper envelope, hence the node at $N=1$ holds the 
value $1$. The inner loop can only iterate once, after which it will add a 
second parabola to the upper envelope; therefore the value of the node at $N=2$ 
is $2$. Now we have $2$ parabolas in the upper envelope. The inner loop can 
either (a) iterate once, in which case the upper envelope will be modified such 
that the parabola at the new grid point will replace the old parabola, and the 
number of parabolas in the upper envelope will remain $2$, or (b) iterate 
twice, and add one parabola to the upper envelope. We construct this tree, 
enumerating all possibilities. At any iteration $N$, assuming that each of these
enumerated situations is equally likely, we compute the average 
number of iterations of the inner loop by summing over the edge weights, and 
dividing by the number of edges. For instance, when $N=4$, the average number 
of iterations of the inner loop is $\frac{1+2+1+2+3+1+2+1+2+3+1+1+3+4}{14}$.

To compute the average number of inner loop iterations for an arbitrary $N$, we 
need to know (a) the sum of edge weights and (b), the number of edges at the 
$N$-th level of the tree. Our tree construction is closely related to a well 
known construction in combinatorial mathematics called the Catalan Family 
Tree~\citep{CatalanTree}. Using results from~\citep{CatalanTree},
(a) the number of edges at the $N$-th level of the tree is a Catalan number, 
given by $\frac{(2N)!}{N!(N+1)!}$, and (b) the sum of edges at the $N$-th 
level is {\em number of $(N+1)$-th generation vertices in the tree of sequences 
with unit increase labeled by 2}, given by $\frac{3(2N)!}{(N+2)!(N-1)!}$ 
The average number of inner loop iterations for an arbitrary $N$ is therefore 
given by $\frac{3N}{N+2}$ which is $O(1)$. The overall runtime is, therefore, 
$O(N\cdot1 + N) = O(N)$.

\subsection{Arbitrary Dimensions}
For the $2D$ grid case, we rewrite Equation 
\ref{eqn:maxdistancetransform2D} as
\begin{equation}
	D(x,y) = \max_{p} \alpha(p-x)^2 + \beta(p-x) + \max_{q} I(p,q) +
\gamma(q-y)^2 + \delta(q-y)\,.
	\label{eqn:2Dto1D}
\end{equation}
Since the last two terms do not depend on $p$, we can first solve the $1D$ grid  problem over $q$, reducing the problem to
\begin{equation}
	D(x,y) = \max_{p} \alpha(p-x)^2 + \beta(p-x) + I'(p,q)\,.
	\label{eqn:reduced1D}
\end{equation}
which is a problem on the $1D$ grid $p$.

Therefore, a maximum distance transform in $2D$ can be computed by first 
computing the maximum distance transform along a column of the grid, and then 
computing the maximum distance transform along each row of the result. This can 
be extended to arbitrary dimensions, by processing the dimensions in any order. 
Kindly note that the solution does not depend on the order in which we process 
the dimensions, as is evident from Equations \ref{eqn:2Dto1D}-\ref{eqn:reduced1D}.
\section{Applications}
The maximum distance transform can be used to maximize quadratic functions on grids of arbitrary dimensions. It can be employed, for instance, in DPMs
for finding the configurations where a score function with quadratic pairwise terms is maximized.

As described in section \ref{section:duality}, the maximum distance transform algorithm can be used to find the minimum distance
transform for downward opening parabolas. In this capacity, it can be used for efficient message passing in graph inference problems
wherever the energy function has negative quadratic pairwise terms. 

Our algorithm, and the one in \citep{Felzdt} together allow us to find the global optimal solution for any inference problem on quadratic functions.
This would allow us to have fewer constraints on our model parameters, and would lead to a better optimal solution.
\section{Conclusion}
In this work, we propose an average time $O(N)$ algorithm for maximizing quadratic functions. We give the proof of correctness 
and analyze the runtime complexity of our algorithm. Given the minimum and maximum distance transform algorithms together, we can efficiently optimise quadratic functions
of any form (without regard to the sign of the quadratic terms), and hopefully this ability will allow us more freedom in the choice of our model parameters in
optimisation problems where efficiency is indispensable.
\bibliography{dt_nips}

\begin{thebibliography}{4}
\providecommand{\natexlab}[1]{#1}
\providecommand{\url}[1]{\texttt{#1}}
\expandafter\ifx\csname urlstyle\endcsname\relax
  \providecommand{\doi}[1]{doi: #1}\else
  \providecommand{\doi}{doi: \begingroup \urlstyle{rm}\Url}\fi

\bibitem[Felzenszwalb et~al.(2010)Felzenszwalb, Girshick, McAllester, and
  Ramanan]{dpmOrig}
P.~F. Felzenszwalb, R.~B. Girshick, D.~McAllester, and D.~Ramanan.
\newblock Object detection with discriminatively trained part based models.
\newblock \emph{IEEE Transactions on Pattern Analysis and Machine
  Intelligence}, 32\penalty0 (9):\penalty0 1627--1645, 2010.

\bibitem[Zhu and Ramanan(2012)]{Zhu_facedetection}
Xiangxin Zhu and Deva Ramanan.
\newblock Face detection, pose estimation, and landmark localization in the
  wild.
\newblock In \emph{CVPR}, 2012.

\bibitem[Felzenszwalb and Huttenlocher(2004)]{Felzdt}
Pedro~F. Felzenszwalb and Daniel~P. Huttenlocher.
\newblock Distance transforms of sampled functions.
\newblock Technical report, Cornell CS, 2004.

\bibitem[\v{S}uni\'{c}(2003)]{CatalanTree}
Zoran \v{S}uni\'{c}.
\newblock Self-describing sequences and the catalan family tree.
\newblock In \emph{The Electronic Journal of Combinatorics}, 2003.

\end{thebibliography}
\bibliographystyle{unsrtnat}
\end{document}